\documentclass[10pt]{amsart}
\usepackage{amsmath,amsfonts,amssymb,mathrsfs}
\usepackage[all]{xy}
\usepackage{graphicx}
\usepackage{amsthm}
\usepackage{psfrag}

\theoremstyle{plain}
\newtheorem{theorem}{Theorem}[section]

\newtheorem{definition}[theorem]{Definition}

\title[Dual futile cycle]{A proof of bistability for the dual futile cycle}
\author{Juliette Hell}
\address[Juliette Hell]{Institut f\"ur Mathematik\\ 
Freie Universit\"at Berlin \\
Arnimallee 3, D--14195 Berlin, Germany}
\email[Juliette Hell]{jhell@zedat.fu-berlin.de}
\author{Alan D. Rendall}
\address[Alan D. Rendall]{Institut f\"ur Mathematik\\
Johannes Gutenberg-Universit\"at\\
Staudingerweg 9, D-55099 Mainz, Germany}
\email[Alan Rendall]{rendall@uni-mainz.de}

\begin{document}
\begin{abstract}
The multiple futile cycle is an important building block in networks 
of chemical reactions arising in molecular biology. A typical process 
which it describes is the addition of $n$ phosphate groups to a protein.
It can be modelled by a system of ordinary differential equations 
depending on parameters. The special case $n=2$ is called the dual futile 
cycle. The main result of this paper is a proof that there are parameter
values for which the system of ODE describing the dual futile cycle has
two distinct stable stationary solutions. The proof is based on bifurcation 
theory and geometric singular perturbation theory. An important entity
built of three coupled multiple futile cycles is the MAPK cascade. It
is explained how the ideas used to prove bistability for the dual futile
cycle might help to prove the existence of periodic solutions for the 
MAPK cascade.
\end{abstract}

\maketitle

\section{Introduction}

A pattern of chemical reactions frequently encountered in cell biology
is one where $n$ phosphate groups are attached to a protein by reactions
catalysed by one enzyme $E$ (a kinase) and removed again by reactions catalysed 
by another enzyme $F$ (a phosphatase). This is sometimes called a multiple 
futile cycle. An introduction to this type of biological system and how it can 
be modelled mathematically using ordinary differential equations (ODE) 
depending on parameters can be found in \cite{wangsontag08a}. It is proved in 
\cite{wangsontag08a} that for a $n$-fold futile cycle, $n\ge 2$, this system of ODE exhibits 
multistationarity for certain values of the parameters, i.e. that there exist 
several different stationary solutions. Upper and lower bounds for the number 
of stationary solutions as a function of $n$ were also proved in 
\cite{wangsontag08a}. For $n=1$ there is a unique stationary solution and it is 
globally asymptotically stable \cite{angeli06}. The results which follow 
concern the case $n=2$ of this system, the dual futile cycle. In that case 
the maximal number of stationary solutions for any values of the parameters is
three. Note for comparison that in the case $n=3$ while the results of 
\cite{wangsontag08a} only guarantee that the maximal number of stationary 
solutions is between three and five it was recently shown in 
\cite{flockerzi13} that the upper bound is sharp. In other words, for $n=3$ 
there are parameter values for which there exist five stationary solutions.

Beyond the question of the number of stationary solutions it is of great 
interest to obtain information about their stability. This allows conclusions 
about the significance of the stationary solutions for the dynamics of more
general solutions of the system. It is of particular interest to know whether
there exist more than one stable stationary solution for fixed values of the 
parameters. This phenomenon is called bistability. (In ODE describing chemical 
systems there are often preferred invariant affine subspaces, the 
stoichiometric compatibility classes. When talking about fixed parameters it 
is understood that the stoichiometric compatibility class has been 
fixed.)  In  \cite{markevich04} it was concluded using numerical and heuristic 
approaches that there is bistability in the dual futile cycle. To the authors'
knowledge there is no rigorous and purely analytical proof of this 
statement in the literature. The main result of the present paper is a proof
of this type.  

In \cite{wangsontag08a} and \cite{markevich04} the reactions are modelled using
a standard Michaelis-Menten scheme for the catalysis of each reaction and mass
action kinetics for the elementary reactions involved. The resulting system will
be called the MM-MA system (Michaelis-Menten via mass action) in what follows
and it is the system of principal interest here. It is possible, via a
quasistationarity assumption of Michaelis-Menten type, to pass formally to a 
smaller system, called the MM system (Michaelis-Menten) in what follows. The
question of bistability in the latter system has been studied by methods which
are partly numerical and heuristic in \cite{ortega06}.

The strategy used in what follows is to first give a rigorous analytical proof
of bistability in the MM system using bifurcation theory. It is shown that 
there is a generic cusp bifurcation and this implies bistability by 
well-known methods \cite{kuznetsov10}. Then bistability is concluded for the 
MM-MA system with the help of geometric singular perturbation theory
\cite{fenichel79} which gives control over the limiting process from the MM-MA 
system to the MM system and can be thought of as a far-reaching generalization 
of earlier work of Tikhonov (see \cite{wasow65}, Sect. 39) and Hoppensteadt 
\cite{hoppensteadt66}. The relevance of this type of result to the 
quasistationary approximation was pointed out in \cite{heineken67}. A similar 
strategy has been used in \cite{wangsontag08b} to prove that generic solutions
of the system describing the dual futile cycle converge to stationary 
solutions. 

The paper is organized as follows. In the next section the basic equations of 
the model studied in the paper are explained. Sect. \ref{techniques} is a 
concise introduction to some of the main mathematical tools used. In Sect.
\ref{bistabmm} bistability is proved for the Michaelis-Menten system and in
Sect. \ref{bistabma} this is used to obtain a corresponding result for 
the full system. In Sect. \ref{further} some directions in which this 
research could be extended are indicated. These concern the MAPK cascade.
In particular it is discussed how Michaelis-Menten reduction can be applied
in that case and an explicit reduced system is presented.

\section{The basic equations}\label{basic}

Consider a chemical system consisting of a protein $Y$ and the substances $YP$
and $YPP$ obtained by attaching one or two phosphate groups to $Y$. The 
reactions which attach phosphate groups are catalysed by an enzyme $E$ and 
those which remove phosphate groups by an enzyme $F$. It is assumed that the
enzyme $E$ can only add one phosphate group before releasing the substrate. 
This is what is called distributive phosphorylation in contrast to processive 
phosphorylation where more than one phosphate is added during one encounter 
between the enzyme and its substrate. It is also assumed that the phosphate 
groups are added at binding sites in a certain order. This is called 
sequential phosphorylation. It is assumed that dephosphorylation by $F$ has 
corresponding properties and that phosphate groups are removed in the reverse 
order to that in which they are added. These assumptions about the nature of 
the phosphorylation and dephosphorylation processes are common in modelling 
approaches in the literature. Adopting them, the chemical reactions we are 
modelling can be written in the following form where the label $k_i$ denotes 
the reaction constant of reaction $i$. 
\begin{equation*}
\xymatrix{
Y+E \ar@/^1pc/[r]^{k_1}
&  YE \ar@/^1pc/[l]^{k_2} \ar[r]^{k_3}& YP+E \\
&&\\
YP+E \ar@/^1pc/[r]^{k_4} & YPE\ar@/^1pc/[l]^{k_5} \ar[r]^{k_6} & YPP+E\\
&&\\
YPP+F \ar@/^1pc/[r]^{k_7} & YPPF \ar@/^1pc/[l]^{k_8} \ar[r]^{k_9}& YP+F\\
&&\\
YP+F \ar@/^1pc/[r]^{k_{10}} & YPF\ar@/^1pc/[l]^{k_{11}} \ar[r]^{k_{12}} & Y+F
}
\end{equation*}

An important biological example is the case
where $Y$ is the extracellular signal-regulated kinase (ERK), $E$ is the 
MAPK/ERK kinase (MEK) and $F$ is MAPK phosphatase 3 (MKP3). This is part of a 
mitogen activated protein kinase (MAPK)
cascade, a pattern of chemical reactions of key significance in molecular 
biology. Some experimental results on this system are summarized in 
\cite{markevich04} and are as follows. Both MEK and MKP3 act in a distributive 
fashion. Dephosphorylation by MKP3 is sequential but phosphorylation by MEK 
has a random component - one order of phosphorylation predominates but the 
other also occurs. There is some discussion in \cite{markevich04} of the 
effects of this fact on the dynamics. In what follows the analysis is 
restricted to the case where both $E$ and $F$ act sequentially. It is of 
interest to note that if distributive phosphorylation is replaced by 
processive phosphorylation for one or both of the enzymes in this model then 
the existence of more than one stationary solution, and in particular 
bistability, can be ruled out \cite{conradi05}.  

The concentration of a substance $Z$ is denoted by 
$x_Z$. If mass action kinetics are imposed then the 
assumptions made up to this point uniquely determine the evolution equations
for the concentrations of the substances taking part in the reactions. These 
are the free substrates $Y$, $YP$ and $YPP$, the free enzymes $E$ and $F$ and 
the substrate-enzyme complexes $YE$, $YPE$, $YPPF$ and $YPF$. They satisfy the 
following system of nine ordinary differential equations:
\begin{eqnarray}
&&\frac{d x_{Y}}{dt}=-k_1 x_{Y}x_{E}+k_2 x_{YE}+k_{12} x_{YPF},\label{dfc1}\\
&&\frac{d x_{YP}}{dt}=-k_4 x_{YP}x_E-k_{10} x_{YP}x_F+k_3 x_{YE}\label{dfc2}\\
&&+k_5 x_{YPE}+k_9x_{YPPF}+k_{11}x_{YPF},\nonumber\\
&&\frac{d x_{YPP}}{dt}=-k_7 x_{YPP}x_F+k_6x_{YPE}+k_8 x_{YPPF},\label{dfc3}\\
&&\frac{d x_{YE}}{dt}=k_1x_{Y}x_{E}-(k_2+k_3) x_{YE},\label{dfc4}\\
&&\frac{d x_{YPE}}{dt}=k_4 x_{YP}x_E-(k_5+k_6) x_{YPE},\label{dfc5}\\
&&\frac{d x_{YPF}}{dt}=k_{10}x_{YP}x_F-(k_{11}+k_{12})x_{YPF},\label{dfc6}\\
&&\frac{d x_{YPPF}}{dt}=k_7 x_{YPP}x_F-(k_8+k_9) x_{YPPF},\label{dfc7}\\
&&\frac{d x_{E}}{dt}=-k_1 x_{Y}x_{E}-k_4x_{YP}x_E+(k_2+k_3)x_{YE}
+(k_5+k_6)x_{YPE},\label{dfc8}\\
&&\frac{d x_{F}}{dt}=-k_7 x_{YPP}x_F-k_{10}x_{YP}x_F\label{dfc9}\\
&&+(k_8+k_9)x_{YPPF}+(k_{11}+k_{12})x_{YPF}.\nonumber
\end{eqnarray}
The reaction constants $k_i$ are positive numbers. It should be noted that 
apart from the differences in notation these equations are identical to those 
in \cite{markevich04} and \cite{wangsontag08a}. Let
\begin{eqnarray}
&&\bar E=x_E+x_{YE}+x_{YPE},\label{totalE}\\
&&\bar F=x_F+x_{YPF}+x_{YPPF},\label{totalF}\\
&&\bar Y=x_Y+x_{YP}+x_{YPP}+x_{YE}+x_{YPE}+x_{YPF}+x_{YPPF}\label{totalY}.
\end{eqnarray}
These quantities, which are the total concentrations of the enzymes and of
the substrate, are conserved under the evolution. Thus $x_E$ and $x_F$ can be 
expressed in terms of $\bar E$ and $\bar F$ and the concentrations of their 
complexes. In a similar way, $x_{YP}$ can be expressed in terms of $\bar Y$,
$x_{Y}$, $x_{YPP}$ and the concentrations of the complexes involving $Y$. It is 
thus possible to discard the evolution equations (\ref{dfc2}), (\ref{dfc8}) and 
(\ref{dfc9}) and reduce the number of equations from nine to six. The notation
$x_{YP}$ should now be thought of as an abbreviation for the expression of 
this quantity in terms of $\bar Y$ and the unknowns in the remaining evolution
equations. The six equations depend on the three parameters $\bar E$, $\bar F$ and 
$\bar Y$. If a solution of the system of nine equations is given, a solution
of the system of six equations with certain values of the parameters  is obtained. Conversely, given a solution of
the system of six equations and a choice of the constants $\bar E$, $\bar F$ 
and $\bar Y$ a solution of the set of nine equations can be obtained. In order 
to ensure that the latter is positive the constants $\bar E$, $\bar F$ and 
$\bar Y$ must be chosen sufficiently large.  

Next we want to pass to the Michaelis-Menten limit. This can be done by 
introducing rescaled variables. If $Z$ is the concentration of a free enzyme 
or that of a substrate-enzyme complex let $x_Z=\epsilon\tilde x_Z$. Let 
$\bar E=\epsilon\tilde E$ and similarly for $F$. The smallness of  
$\epsilon>0$ amounts to the enzyme concentrations being small compared to the 
concentrations of the substrates. Finally, let $\tau=\epsilon t$. 
For $\epsilon>0$ small, the system exhibits fast-slow dynamics. The two different time scales are reflected in the time variables $t$ and $\tau$. 
The transformed equations are
\begin{eqnarray}
&&\frac{d x_{Y}}{d\tau}=-k_1 x_{Y}\tilde x_{E}+k_2 \tilde x_{YE}
+k_{12} \tilde x_{YPF},\label{rescaled1}\\
&&\frac{dx_{YPP}}{d\tau}=-k_7 x_{YPP}\tilde x_F+k_6\tilde x_{YPE}
+k_8 \tilde x_{YPPF},\label{rescaled3}\\
&&\epsilon\frac{d \tilde x_{YE}}{d\tau}=k_1x_{Y}\tilde x_{E}-(k_2+k_3) 
\tilde x_{YE},\label{rescaled4}\\
&&\epsilon\frac{d \tilde x_{YPE}}{d\tau}=k_4 x_{YP}\tilde x_E
-(k_5+k_6) \tilde x_{YPE},\label{rescaled5}\\
&&\epsilon\frac{d \tilde x_{YPF}}{d\tau}=k_{10}x_{YP}\tilde x_F
-(k_{11}+k_{12})\tilde x_{YPF},\label{rescaled6}\\
&&\epsilon\frac{d \tilde x_{YPPF}}{d\tau}=k_7 x_{YPP}\tilde x_F-(k_8+k_9) 
\tilde x_{YPPF}.\label{rescaled7}
\end{eqnarray}
Note that equations (\ref{dfc2}), (\ref{dfc8}) and (\ref{dfc9}) could also be 
rescaled in a similar way so that in total a system of nine equations 
depending on $\epsilon$ is obtained whose solutions for any $\epsilon>0$ are 
in one to one correspondence with the solutions of the system 
(\ref{dfc1})-(\ref{dfc9}). Setting $\epsilon=0$ in equations 
(\ref{rescaled4})-(\ref{rescaled7}) gives
\begin{eqnarray}
&&\tilde x_{YE}=\frac{k_1}{k_2+k_3}x_{Y}\tilde x_{E},\label{stat1}\\
&&\tilde x_{YPE}=\frac{k_4}{k_5+k_6} x_{YP}\tilde x_E,\label{stat2}\\
&&\tilde x_{YPF}=\frac{k_{10}}{k_{11}+k_{12}}x_{YP}\tilde x_F,\label{stat3}\\
&&\tilde x_{YPPF}=\frac{k_7}{k_8+k_9} x_{YPP}\tilde x_F.\label{stat4}
\end{eqnarray}
Adding these equations in pairs gives
\begin{eqnarray}
&&\tilde E=\left[1+\frac{k_1}{k_2+k_3}x_{Y}
+\frac{k_4}{k_5+k_6} x_{YP}\right]\tilde x_{E}\label{erel},\\
&&\tilde F=\left[1+\frac{k_{10}}{k_{11}+k_{12}}x_{YP}
+\frac{k_7}{k_8+k_9} x_{YPP}\right]\tilde x_F.\label{frel}
\end{eqnarray}
Using (\ref{stat1})-(\ref{stat4}) the evolution equations 
(\ref{rescaled1})-(\ref{rescaled3}) can be rewritten as
\begin{eqnarray}
&&\frac{d x_{Y}}{d\tau}=-\frac{k_1k_3}{k_2+k_3} x_{Y}\tilde x_{E}
+\frac{k_{10}k_{12}}{k_{11}+k_{12}} x_{YP}\tilde x_F,\label{mmraw1}\\
&&\frac{d x_{YPP}}{d\tau}=\frac{k_4k_6}{k_5+k_6}x_{YP}\tilde x_E
-\frac{k_7k_9}{k_8+k_9} x_{YPP}\tilde x_F.\label{mmraw3}
\end{eqnarray} 
Note that  $\bar Y= x_Y+x_{YP}+x_{YPP} + \epsilon ( \tilde x_{YE}+\tilde x_{YPE}+\tilde x_{YPF}+\tilde x_{YPPF})$. We define $\tilde Y=\bar Y(0)$ to emphasize the dependence on $\epsilon$.  Then for $\epsilon=0$ 
the relation 
\begin{equation}\label{totalY2}
\tilde Y=x_Y+x_{YP}+x_{YPP}
\end{equation}
holds so that setting $x_{YP}=\tilde Y-x_Y-x_{YPP}$ this is a closed system of 
equations. Summing up, the equations are of the form
\begin{eqnarray}
&&\frac{d x_{Y}}{d\tau}=-v_1+v_2,\\
&&\frac{d x_{YPP}}{d\tau}=v_3-v_4
\end{eqnarray}  
where, using (\ref{erel}) and (\ref{frel}),
\begin{eqnarray}
&&v_1=\frac{a_1x_Y}{1+b_1x_Y+c_1x_{YP}},\\
&&v_2=\frac{a_2x_{YP}}{1+c_2x_{YP}+d_2x_{YPP}},\\
&&v_3=\frac{a_3x_{YP}}{1+b_1x_Y+c_1x_{YP}},\\
&&v_4=\frac{a_4x_{YPP}}{1+c_2x_{YP}+d_2x_{YPP}}
\end{eqnarray}
and
\begin{eqnarray}
&&a_1=\frac{k_1k_3\tilde E}{k_2+k_3},\ \ \
a_2=\frac{k_{10}k_{12}\tilde F}{k_{11}+k_{12}},\\
&&a_3=\frac{k_4k_6\tilde E}{k_5+k_6},\ \ \
a_4=\frac{k_7k_9\tilde F}{k_8+k_9},\\
&&b_1=\frac{k_1}{k_2+k_3},\ \ \
c_1=\frac{k_4}{k_5+k_6},\\
&&c_2=\frac{k_{10}}{k_{11}+k_{12}},\ \ \
d_2=\frac{k_7}{k_8+k_9}.
\end{eqnarray}
The conditions for a stationary solution are $v_1=v_2$ and $v_3=v_4$. From 
now on the simplifying assumption will be made that
\begin{equation}\label{common}
b_1=c_1=c_2=d_2.
\end{equation} 
Call the common value of these quantities $b$. 
The simplifying assumption (\ref{common}) means that the ratios $\frac{k_i}{k_{i+1}+k_{i+2}}$, $i=1,4,7,10$, between the constants of the reactions producing and consuming the substrate-enzyme complexes during phosphorylation are equal. 
Note that (\ref{common}) is
equivalent to a set of relations between the $k_i$ and $b$. If $\tilde E$ and 
$\tilde F$ are given then $k_3$, $k_6$, $k_9$ and $k_{12}$ are determined
by the $a_i$ and $b$. There are then four degrees of  freedom in choosing the other
$k_i$ so as to satisfy the other relations. Thus if values $a_i$, $b$, 
$\tilde E$ and $\tilde F$ are given they can be realized by many choices of
$k_i$. 

\section{Review of some mathematical techniques}\label{techniques}

In this section some ideas from bifurcation theory and geometric singular 
perturbation theory will be reviewed. Bifurcation theory for 
ordinary differential equations concerns parameter-dependent systems
$\dot x=f(x,\mu)$ where the variable $x$ and the parameter $\mu$ belong to subsets of 
Euclidean spaces. Under some circumstances the flows of the equations
for different values of $\mu$ are related to each other by a diffeomorphism 
of the space with coordinates $x$. In that case the qualitative nature of the 
flow does not change with $\mu$. This is for instance true in a 
neighbourhood of a point $x_0$ and a parameter value $\mu_0$ for which 
$f(x_0,\mu_0)\ne 0$ (Flow-box Theorem), or in a neighbourhood of a hyperbolic equilibrium (Hartman-Grobman Theorem). If the flows cannot be related in this way for 
$\mu$ close to some value $\mu_0$ then it is said that a bifurcation
occurs at $\mu_0$.
For example at an equilibrium $x_0$ (where of course $f(x_0,\mu_0)=0$), a bifurcation occurs iff the linearisation $D_xf(x_0,\mu_0)$ admits eigenvalues with vanishing real part. In this case a non-trivial centre manifold exists. 
 The aim of bifurcation theory is to obtain insights into 
the qualitative changes in the flow in this kind of situation. This can
often provide valuable information on the dynamics of the system for 
parameter values close to $\mu_0$, for instance the number and 
stability of stationary solutions. This information is obtained when
certain criteria are satisfied. Typical examples of relevant conditions are 
whether certain combinations of derivatives of $f$ at $(x_0,\mu_0)$ are 
zero or non-zero or the signs of such combinations. When studying a 
bifurcation at a stationary point $x_0$ it is often possible to reduce the
dimension of the space of unknowns $x$ using centre manifold theory. A
detailed discussion of this can be found in chapter 5 of \cite{kuznetsov10}. 
The essential qualitative features occur on the centre manifold whose dimension
is equal to the dimension of the space of generalized eigenvectors of the
linearization $Df(x_0,\mu_0)$ corresponding to purely imaginary eigenvalues. 

To make this more concrete we consider the example of greatest importance 
in what follows, the cusp bifurcation. Suppose that the dynamical system
is one-dimensional, remembering that this may have resulted from a system
of higher dimension by centre manifold reduction. Suppose that $\mu$
has two components $\mu_1$ and $\mu_2$. By choosing coordinates 
appropriately it can be assumed that the bifurcation occurs at the point
$(0,0)$. The system is said to exhibit a generic cusp bifurcation if 
\begin{equation}\label{cusp}
f(0,0)=0, f_x(0,0)=0, f_{xx}(0,0)=0, f_{xxx}(0,0)\ne 0,
f_{x\mu_1}f_{\mu_2}-f_{x\mu_2}f_{\mu_1}\ne 0 
\end{equation}
where the subscripts
denote partial derivatives with respect to the corresponding variables.
If only the first three of these conditions are satisfied it might be said 
loosely that there is a cusp bifurcation but this does not exclude the 
possibility that extra degeneracies might occur. These can be ruled out by
imposing the last two conditions and this is what is meant by the term 
'generic' here. When the one-dimensional system arises from a system of 
higher dimension by centre manifold reduction the diagnostic conditions for
a cusp bifurcation can be reexpressed in terms of algebraic conditions on
derivatives of the original system, as explained in chapter 8 of 
\cite{kuznetsov10}. This will be seen in more detail in the example of the 
dual futile cycle in Sect. \ref{bistabmm}. The significance of the cusp for
the results of this paper is that when a bifurcation of this type with
$f_{xxx}(0,0)<0$ occurs at some point there are nearby parameter values for
which there exist two stable stationary solutions, and one which is unstable. 
The case $f_{xxx}(0,0)>0$
gives one stable and two unstable stationary solutions and is not 
relevant for what follows.  
The name cusp bifurcation comes from the following. Figure \ref{cuspfig} shows $\mathbb{R}^3$, where the horizontal plane is the plane of parameters $\mu=(\mu_1, \mu_2)$, and the vertical direction the one dimensional variable $x$.
\psfrag{1mu}{$\mu_1$}
\psfrag{2mu}{$\mu_2$}
\psfrag{x}{$x$}
\begin{figure}[h]
\includegraphics[width=1\textwidth]{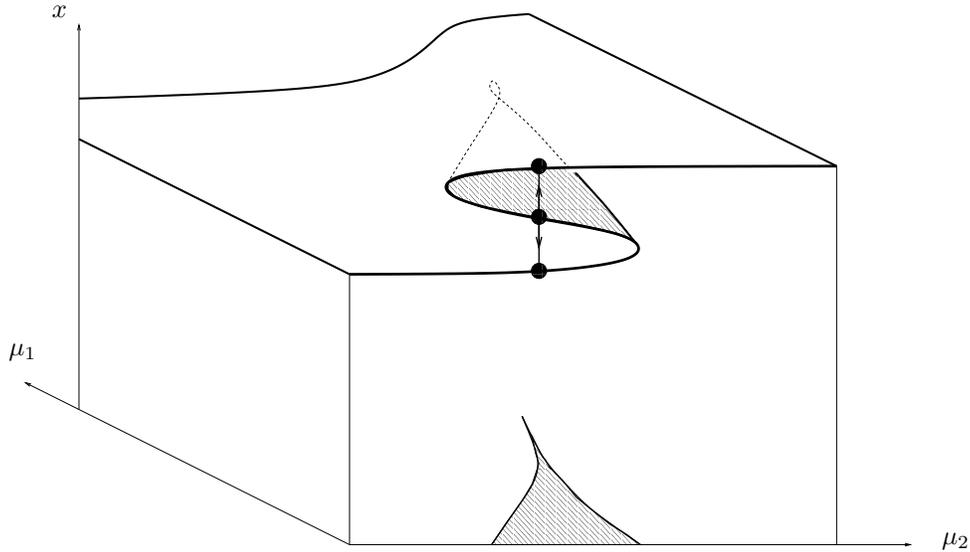}
\caption{ \label{cuspfig} Bifurcation diagram for a generic cusp bifurcation: the unstable branch of the surface of equilibria is shaded, as well as  the region in the parameter plane with multiple equilibria.}
\end{figure}
The set of equilibria forms a surface in this bifurcation diagram. Near a cusp bifurcation, this surface shows a fold. The  line delimiting the folding region projected  into the horizontal parameter plane forms a cusp, although the surface
in $\mathbb{R}^3$ is smooth. For a given parameter in the cusp region, there are two stable equilibria and one unstable one. Furthermore there are heteroclinic orbits connecting the unstable equilibrium to the two stable ones.  

Next some remarks will be made on geometric singular perturbation theory
(GSPT). This theory was developed by Fenichel \cite{fenichel79} and section
3 of that paper provides an introduction to some of the main ideas involved.
The standard situation in which this theory is applied is for a system of 
ODE of the form
\begin{equation}\label{gspt}
\begin{cases}
x'=f(x,y,\epsilon),\\
\epsilon y'=g(x,y,\epsilon),
\end{cases}
\end{equation}
where the prime denotes the derivative with respect to a time coordinate
$\tau$. Here $f$ and $g$ are 
smooth functions and $\epsilon$ is a parameter. The aim is to understand the 
qualitative behaviour of solutions of this system in the limit 
$\epsilon\to 0$ where the system (\ref{gspt}) reads
\begin{equation}\label{gspt0}
\begin{cases}
x'=f(x,y,0),\\
0=g(x,y,\epsilon).
\end{cases}
\end{equation}
The system (\ref{gspt0}) is called the reduced system. 
The dynamics of the reduced system  may provide useful information about the 
qualitative behaviour of solutions of the system (\ref{gspt}) with $\epsilon$ small and
non-zero. The system (\ref{gspt}) is singular at $\epsilon=0$ since the time derivative
of $y$ is multiplied by a factor which vanishes there. This means that 
regular perturbation theory cannot be applied. Note that the system 
(\ref{rescaled1})-(\ref{rescaled7}) is of this general form.

Suppose that in the above system $x$ is a point of ${\mathbb R}^{n_1}$ and $y$ a 
point of ${\mathbb R}^{n_2}$. Let $\tau=\epsilon t$ for $\epsilon>0$ 
and denote the derivative with respect to $t$ by a dot. Transforming the
equations to the time coordinate $t$ and adding the equation $\dot\epsilon=0$
gives the following  system of $n_1+n_2+1$ equations
\begin{equation}\label{extsys}
\begin{cases}
\dot x=\epsilon f(x,y,\epsilon),\\ 
\dot y=g(x,y,\epsilon),\\
\dot \epsilon=0,
\end{cases}
\end{equation}
Call system (\ref{extsys}) the extended system. This extends smoothly to 
$\epsilon=0$ to 
\begin{equation}\label{extsys0}
\begin{cases}
\dot x=0,\\ 
\dot y=g(x,y,0),\\
\dot\epsilon=0,
\end{cases}
\end{equation}
 Assume that the set defined by the equation 
$g(x,y,0)=0$ is equivalent to $y=h_0(x)$ for a smooth function $h_0$. In 
particular it is a manifold $M_0$. This manifold consists of stationary points 
of the extended system (\ref{extsys}) at $\epsilon=0$. Assume further that the linearization of 
system (\ref{extsys}) at any point of $M_0$ has no purely imaginary eigenvalues other than 
the zero eigenvalues arising from the fact that $M_0$ consists of stationary
solutions and the fact that $\epsilon$ is a conserved quantity. 
Call these eigenvalues corresponding to eigenvectors transverse
to $M_0$ the transverse eigenvalues. They correspond to the eigenvalues
of the linearization of the system $\dot y=g(x,y,0)$ with $x$ held fixed.
This leads to the following
\begin{definition} \label{transeigen}
For a point $(x,h_0(x))\in M_0$, we define the  transverse eigenvalues at $x$ as the eigenvalues of the linearisation $D_yg(x,h_0(x),0)$, under the assumption that none of them is purely imaginary. 
\end{definition}
A centre manifold $M$ of the extended system (\ref{extsys}) at any  point of $M_0$ contains 
$M_0$. It was shown in \cite{fenichel79} that there exists a manifold $M$ 
which is a centre manifold for all points of $M_0$ close to a given point. 
The manifold $M$  is sometimes called a slow manifold. The restriction of the extended 
system (\ref{extsys}) to the slow manifold $M$ can be interpreted as a system of $n_1$ equations 
depending on the parameter $\epsilon$. Remarkably, rewriting this system in 
terms of the time coordinate $\tau$ gives a system which depends on $\epsilon$ in 
a regular fashion, even at $\epsilon=0$. For $\epsilon=0$ it coincides with 
the reduced system (\ref{gspt0}). For $\epsilon>0$ call this the perturbed 
reduced system.    

The reduction theorem of Shoshitaishvili (see \cite{kuznetsov10}, Theorem 5.4)
shows that the qualitative behaviour of solutions of the extended system (\ref{extsys}) near a 
point of $M_0$ is determined in a simple way by the dynamics on $M$. In fact 
it is topologically equivalent to the product of the dynamics on $M$ with a 
standard saddle. What this means in
practise is that dynamical features of the reduced system (\ref{gspt0}) are inherited by
the full system (\ref{gspt}) for $\epsilon$ small. In 
the case of the central example of this paper the real parts of the 
transverse eigenvalues of the linearization at points of $M_0$ defined in \ref{transeigen} are all 
negative. Thus the standard saddle mentioned above is a hyperbolic 
sink in this case. This implies that bistability of the MM system (which is 
the reduced system in this case) is inherited by the perturbed reduced system. 
It then follows from the theorem of Shoshitaishvili that bistability is 
inherited by the MM-MA system. In other situations, as will be explained 
further in Sect. \ref{further}, the existence of periodic solutions of the 
reduced system is inherited by the full system.

\section{Bistability for the Michaelis-Menten system}\label{bistabmm}

The MM system can be written as
\begin{eqnarray}
&&\frac{d x_{Y}}{d\tau}=-\frac{a_1x_Y}{1+b(\tilde Y-x_{YPP})}
+\frac{a_2x_{YP}}{1+b(\tilde Y-x_Y)},\label{mm1}\\
&&\frac{d x_{YPP}}{d\tau}=\frac{a_3x_{YP}}{1+b(\tilde Y-x_{YPP})}
-\frac{a_4x_{YPP}}{1+b(\tilde Y-x_Y))}.\label{mm2}
\end{eqnarray}
It has been 
observed in \cite{ortega06} that if certain restrictions are imposed on the 
parameters it is possible to find an explicit stationary solution of  
equations (\ref{mm1})-(\ref{mm2}) with interesting properties. Consider 
a stationary solution that we call $B$, with starred coordinates $(x_Y^*, x_{YPP}^*)$ and the corresponding quantity $x_{YP}^*=\tilde{Y}- x_Y^*-x_{YPP}^*$.  We look for an equilibrium $B$ satisfying $x_Y^*=x_{YPP}^*$, a condition which is
satisfied by the special solutions considered in \cite{ortega06}. Then the 
denominators in the expressions for the evolution equations are all equal and
the equations for stationary solutions imply that
\begin{equation}
a_1 x_Y^*=a_2 x_{YP}^*,\ \ \ a_3 x_{YP}^*=a_4 x_{YPP}^*.
\end{equation}
It follows that 
\begin{equation}\label{cond1}
\frac{a_2a_4}{a_1a_3}=1.
\end{equation}
Let $N=1+b(\tilde Y-x_Y^*)$. The linearization of the system at the stationary 
point $B$ is $N^{-2}$ times
\begin{equation}
\left[
{\begin{array}{cc}\label{rescaledmatrix}
-(a_1+a_2)+b[-a_1(\tilde Y-x)-a_2x]& -a_1bx-a_2 N 
\\ -a_3 N-ba_4x& -(a_3+a_4)+b[-a_4(\tilde Y-x)-a_3x]
\end{array}}
\right]
\end{equation}  
where $x=x_Y^*$. The equations $a_1 x=a_2 (\tilde Y-2x)$ and 
$a_4 x=a_3 (\tilde Y-2x)$ can be used to eliminate $a_1$ and $a_4$ in favour of
$a_2$ and $a_3$. After simplification and multiplication by an overall factor
$x$ the matrix (\ref{rescaledmatrix}) becomes
\begin{equation} \label{simprescmatrix}
\left[
{\begin{array}{cc}
-(\tilde Y-x)a_2-ba_2(\tilde Y^2-3x\tilde Y+3x^2)& -a_2 x+a_2 b(-2\tilde Y x
+3x^2)
\\ -a_3 x+a_3 b(-2\tilde Y x+3x^2)&-(\tilde Y-x)a_3
-ba_3(\tilde Y^2-3x\tilde Y+3x^2) 
\end{array}}
\right].
\end{equation} 
Now we look for a stationary point $B$ where a bifurcation takes place, so that the above matrix should admit a zero eigenvalue. 
The determinant of this matrix vanishes precisely when
\begin{equation}
|b(\tilde Y-x)+b^2(\tilde Y^2-3x\tilde Y+3x^2)|=|bx-b^2 (-2\tilde Y x+3x^2)|.
\end{equation}
There are two cases according to the relative signs of the quantities
inside the absolute values. Setting $u=b\tilde Y$ and $v=bx$ we get the 
alternative conditions $(u-3v+1)(u-2v)=0$ and $u(u-v+1)=0$. The case $u=0$
can be discarded since it corresponds to a vanishing concentration of 
substrate. The case $u=2v$ is also not relevant since it corresponds to
a vanishing concentration of $x_{YP}$. The case $u=v-1$ leads to a negative 
concentration of $x_{YP}$. Consider now the remaining case $u=3v-1$. A
computation shows that 
\begin{equation}\label{expstat}
x_Y^*=x_{YPP}^*=\frac{b\tilde Y+1}{3b},\ \ \  x_{YP}^*=\frac{b\tilde Y-2}{3b}.
\end{equation}
Now 
\begin{equation}\label{cond2}
\frac{b\tilde Y-2}{b\tilde Y+1}=\frac{x_{YP}^*}{x_Y^*}=\frac{a_1}{a_2}
=\sqrt{\frac{a_1a_4}{a_2a_3}}.
\end{equation}
We have a bifurcation if and only this last relation is satisfied. A necessary 
and sufficient condition that it can be satisfied for parameters with 
$b\tilde Y$ positive is that $\frac{a_2a_3}{a_1a_4}>1$. In that case 
$b\tilde Y>2$. We summarize the conditions necessary and sufficient for the existence of  the  stationary point denoted by $B$ with coordinates given by (\ref{expstat}):
\begin{eqnarray} 
\frac{a_2a_4}{a_1a_3}=1  & \text{ equilibrium condition}   \label{equilcond}\\
1> \frac{b\tilde{Y}-2}{b\tilde{Y}+1} = \sqrt{ \frac{a_1a_4}{a_2a_3}  } >0 &\text{bifurcation condition} \label{bifcond}
\end{eqnarray}
 Substituting the 
bifurcation condition into the linearization $N^{-2}$ times (\ref{rescaledmatrix}) simplifies it to 
\begin{equation} \label{linearization}
-\frac{b\tilde Y}{N^2}\left[
{\begin{array}{cc}
a_2&a_2
\\a_3&a_3
\end{array}}
\right].
\end{equation}
This matrix has rank one with the right eigenvector corresponding to the zero 
eigenvalue being the transpose of $[1,-1]$ and the left eigenvector 
$[a_3,-a_2]$. The other eigenvalue has the same sign as the trace and is 
negative. 

Next it will be shown that there is a non-degenerate cusp bifurcation at the 
point $B$. The defining conditions for a bifurcation of this type are given
by (\ref{cusp}) in the case of a one-dimensional system. As discussed in
Sect. \ref{techniques} a non-degenerate cusp bifurcation in higher dimensional
dynamical systems is defined by relating the given system to the 
one-dimensional case by means of centre manifold reduction. The 
essential dynamics is determined by that on the centre manifold which is 
one-dimensional in this case. The point $B$ in the MM
system has a one-dimensional centre manifold and so these techniques are
applicable. The reduction theorem makes transparent what is happening on an
abstract level. On the other hand the concrete calculations which are needed
to verify the presence of a non-degenerate cusp bifurcation by obtaining
a suitable approximation to the centre manifold and analysing the 
dynamics on that manifold are hard to keep track of. Here we follow
an approach to organizing these calculations explained in Section 8.7 of
\cite{kuznetsov10}.

It is necessary to calculate certain combinations of derivatives of the 
right hand side of the MM system at the point $B$ with respect to the 
unknowns and the parameters. For the cusp bifurcation it is necessary to
vary two parameters. We choose these to be $a_1$ and $a_4$ and hold all other
parameters in the system fixed. It is convenient to introduce the notation
$X^i$, $i=1,2$ for the right hand sides of the equations of the MM system
(\ref{mm1})-(\ref{mm2}) and $x_i$, $i=1,2$ for the coordinates $x_{Y}$ and 
$x_{YPP}$. The components of the right and left eigenvectors of the 
linearization will be denoted by $R^i$ and $L_i$, respectively. For 
derivatives we use the following abbreviations
\begin{equation}
X^i{}_{,j}=\frac{\partial X^i}{\partial x_j},
X^i{}_{,jk}=\frac{\partial^2 X^i}{\partial x_j\partial x_k},
X^i{}_{,jkl}=\frac{\partial X^i}{\partial x_j\partial x_k\partial x_l}.
\end{equation}

Now higher derivatives of the coefficients in the system MM will be examined.
For this it is useful to consider derivatives of rational functions more 
generally. Each $X^i$ is a sum of rational functions where numerator and 
denominator are linear in the arguments. If $f$ and $g$ are smooth functions 
of two variables with $g$ non-vanishing then
\begin{eqnarray}
&&D_i\left(\frac{f}{g}\right)=\left(\frac{D_if}{g}\right)
+fD_i\left(\frac{1}{g}\right), i=1,2,\\
&&D_iD_j\left(\frac{f}{g}\right)=\left(\frac{D_iD_jf}{g}\right)
+D_ifD_j\left(\frac{1}{g}\right)\\
&&+D_jfD_i\left(\frac{1}{g}\right)
+fD_iD_j\left(\frac{1}{g}\right).\nonumber
\end{eqnarray}
When $f$ is linear the term with $D_iD_jf$ vanishes. Now
\begin{eqnarray}
&&D_i\left(\frac{1}{g}\right)=-\frac{D_i g}{g^2},\\
&&D_iD_j\left(\frac{1}{g}\right)=-\frac{D_iD_jg}{g^2}+2\frac{D_igD_jg}{g^3}.
\end{eqnarray}
When $g$ is linear the first term vanishes. Thus when both $f$ and $g$ are 
linear we get 
\begin{equation}
D_iD_j\left(\frac{f}{g}\right)=\frac{-gD_ifD_jg-gD_igD_jf+2fD_igD_jg}{g^3}.
\end{equation}
Contracting with $R^iR^j$ (i.e. multiplying by this expression and summing over
repeated indices) gives
\begin{equation}
R^iR^jD_iD_j\left(\frac{f}{g}\right)=
2\left(\frac{-g(R^iD_if)(R^jD_jg)+f(R^iD_ig)^2}{g^3}\right).
\end{equation}   
Here and in what follows we use the summation convention - sums over repeated
indices are implicitly assumed.

Consider the quantities $X^i{}_{,jk}R^jR^k$, evaluated at $B$. They are given 
by
\begin{eqnarray}
&&X^1{}_{,jk}R^jR^k=2N^{-3}[a_1b(-N+bx_Y)+a_2b^2x_{YP}],\\
&&X^2{}_{,jk}R^jR^k=2N^{-3}[a_3b^2x_{YP}+a_4b(-N+bx_{YPP})].
\end{eqnarray}    
After some substitutions this gives
\begin{eqnarray}
&&X^1{}_{,jk}R^jR^k=\frac23 N^{-3}a_2[b(b\tilde Y-2)+3b^2x_{YP}]
=4N^{-3}a_2 b^2x_{YP},\\
&&X^2{}_{,jk}R^jR^k=\frac23 N^{-3}a_3[3b^2x_{YP}+b(b\tilde Y-2)]
=4N^{-3}a_3 b^2x_{YP}.
\end{eqnarray}
It follows that $L_iX^i{}_{,jk}R^jR^k=0$ and this is one of the conditions for a
cusp bifurcation.

The next step is to examine the third order derivatives. Note first that
\begin{equation}
D_iD_jD_k\left(\frac{1}{g}\right)
=-\frac{6D_igD_jgD_kg}{g^4}.
\end{equation}
In the given situation
\begin{eqnarray}
&&D_iD_jD_k\left(\frac{f}{g}\right)\\
&&=\frac{2(D_ifD_jgD_kg+D_kfD_igD_jg+D_jfD_kgD_ig)}{g^2}
-\frac{6fD_igD_jgD_kg}{g^4}\nonumber
\end{eqnarray}
and
\begin{equation}
K^iK^jK^kD_iD_jD_k\left(\frac{f}{g}\right)
=\frac{6(K^jD_jg)^2(gK^iD_if-fK^iD_ig)}{g^4}.
\end{equation}
It follows that
\begin{equation}
K^iK^jK^kD_iD_jD_k\left(\frac{f}{g}\right)
=-3\frac{K^iD_ig}{g}K^iK^jD_iD_j\left(\frac{f}{g}\right).
\end{equation}
Hence
\begin{equation}\label{thirdder}
L_iX^i{}_{,jkl}R^jR^kR^l=-12N^{-4}(a_2^2+a_3^2)b^3x_{YP}.
\end{equation}
In the one-dimensional case one of the conditions for a cusp is that the third
derivative of the right hand side does not vanish. In higher dimensions the
analogue of the third derivative is not just the expression in (\ref{thirdder}).
There is an extra correction given on p. 374 of \cite{kuznetsov10}. To compute 
this we need a vector whose image under the linearization at $B$ is the
vector with components $X^i{}_{,jk}R^jR^k$. Call it Z. This is 
$-4N^{-1}(a_2+a_3)^{-1}\tilde Y^{-1} bx_{YP}[a_2,a_3]^T$. The crucial quantity is
\begin{equation}\label{crucial}
X^i{}_{,jkl}R^jR^kR^lL_i-3X^i{}_{,jk}R^jZ^kL_i.
\end{equation}
To evaluate this we first calculate the second derivatives of $X^1$ and $X^2$,
which are
\begin{equation}N^{-3}\left[
{\begin{array}{cc}
-2Na_2b+2a_2b^2x_{YP}& N(a_1-a_2)b
\\ N(a_1-a_2)b&2b^2a_1x_Y
\end{array}}\right]
\end{equation}
and
\begin{equation}N^{-3}\left[
{\begin{array}{cc}
2b^2a_4x_{YPP}& N(-a_3+a_4)b
\\ N(-a_3+a_4)b&-2Na_3b+2a_3b^2x_{YP}
\end{array}}\right].
\end{equation}
Contracting with $R$ corresponds to taking the differences of the columns
of these matrices. The results are
\begin{equation}
N^{-3}[-N(a_1+a_2)b+2a_2b^2x_{YP}\ \ \ N(a_1-a_2)b-2b^2a_1x_Y]
\end{equation}
and 
\begin{equation}
N^{-3}[2b^2a_4x_{YPP}+N(a_3-a_4)b\ \ \ N(a_3+a_4)b-2a_3b^2x_{YP})].
\end{equation}
Eliminating $a_1$ and $a_4$ in terms of $a_2$ and $a_3$ as has been done
in previous arguments and expressing everything in terms of $b\tilde Y$ gives
\begin{equation}
N^{-3}\left[-\frac23 a_2b(b\tilde Y+1)\ \ \ -\frac23 a_2b(b\tilde Y+1)\right],
\ \ \
N^{-3}\left[\frac23 a_3b(b\tilde Y+1)\ \ \ \frac23 a_3b(b\tilde Y+1)\right].
\end{equation}
Contracting this with $[a_2,a_3]^T$ and $L$ gives 
$-4N^{-3}a_2 a_3 (a_2+a_3)b^2x_Y$. Hence
\begin{equation}
X^i{}_{,jk}R^jZ^kL_i=16N^{-3}a_2 a_3 b^2x_{YP}\tilde Y^{-1}.
\end{equation}
The crucial quantity is strictly negative and so another of the conditions for
the cusp bifurcation is satisfied.

It remains to examine the derivatives with respect to parameters. Recall that
the parameters to be varied are $\mu_1=a_1$ and $\mu_2=a_4$. The derivatives 
with respect to the two parameters of the right hand sides of the two MM 
equations are
\begin{equation}
\left(\frac{-x_Y}{1+b(x_Y+x_{YP})},0\right),\ \ \ 
\left(0,\frac{-x_{YPP}}{1+b(x_{YP}+x_{YPP})}\right).
\end{equation}
Contracting these with respect to the left eigenvector $L$ gives
\begin{equation}
\frac{-a_2x_Y}{1+b(x_Y+x_{YP})},\ \ \ \frac{a_3x_{YPP}}{1+b(x_{YP}+x_{YPP})}.
\end{equation}
Differentiating the vector field in the direction of the right eigenvector $R$
gives
\begin{equation}
\left(\frac{-a_1(1+bx_{YP})-a_2bx_{YP}}{[1+b(x_Y+x_{YP})]^2},
\frac{a_4(1+bx_{YP})+a_3bx_{YP}}{[1+b(x_{YP}+x_{YPP})]^2}\right).
\end{equation}
The derivatives of this with respect to the two parameters are
\begin{equation}
\left(\frac{-(1+bx_{YP})}{[1+b(x_Y+x_{YP})]^2}
,0\right),\ \ \
\left(0,\frac{(1+bx_{YP})}{[1+b(x_{YP}+x_{YPP})]^2} 
\right).
\end{equation}
Contracting these with respect to the left eigenvector gives
\begin{equation}
\frac{-a_2(1+bx_{YP})}{[1+b(x_{YP}+x_{YPP})]^2},\ \ \ 
\frac{-a_3(1+bx_{YP})}{[1+b(x_{YP}+x_{YPP})]^2}.
\end{equation}
It follows that
\begin{eqnarray}
&&\left[\frac{\partial X^i}{\partial\mu_1}\frac{\partial X^j{}_{,k}}
{\partial\mu_2}
-\frac{\partial X^i}{\partial\mu_2}\frac{\partial X^j{}_{.k}}{\partial\mu_1}
\right]L_iL_jR^k\nonumber\\
&&=-(a_2^2+a_3^2)N^{-3}x_Y(1+bx_{YP}).
\end{eqnarray} 
This is the last diagnostic quantity for the cusp bifurcation and it is seen 
to be non-zero. This means that all conditions needed for the proof that there 
is a generic cusp bifurcation at $B$ have been verified.

\begin{theorem} \label{bistabMM}
There exist values of the parameters $a_i$, $b$ and $\bar Y$
for which the  MM system \eqref{mm1}-\eqref{mm2} has three stationary solutions, 
two of which are locally asymptotically stable and one of which is a saddle of Morse index one. Heteroclinic orbits connect the saddle to both stable equilibria.   
\end{theorem}

\begin{proof}
It has been shown that when the conditions (\ref{cond1}) and 
(\ref{cond2}) are assumed the system exhibits a generic cusp bifurcation. It 
follows from the analysis in Section 8.7 of \cite{kuznetsov10} that centre 
manifold reduction leads to a one-dimensional system satisfying the hypotheses 
of Theorem 8.1 of \cite{kuznetsov10}. This implies that near the bifurcation 
point the dynamics on the centre manifold is topologically equivalent to the 
model system for the cusp bifurcation. Combining this with Theorem 5.4 of 
\cite{kuznetsov10} and using the fact that the non-zero eigenvalue of the 
linearization at the bifurcation point is negative gives the statement about 
the qualitative properties of the stationary solutions. 
\end{proof}
\section{Bistability for the MM-MA system}\label{bistabma}

In this section it is shown that Theorem \ref{bistabMM}, which asserts bistability for the 
MM system, can be used to show an analogous result for the MM-MA system. 
The existence of three stationary solutions for the dual futile cycle with 
suitable parameter values was proved in \cite{wangsontag08a}. The aim here is 
to obtain information about the stability of these solutions.

\begin{theorem} \label{bistabMM-MA}
There exist values of the parameters $k_i$ and values of the
conserved quantities $\bar E$, $\bar F$ and $\bar Y$ for which the system 
(\ref{dfc1})-(\ref{dfc9}) has three stationary solutions, two of which are 
stable and one of which is a saddle of Morse index one. 
Their coordinates are compatible with the given  values of the 
conserved quantities.
Furthermore heteroclinic orbits connect the saddle to both stable equilibria.   

\end{theorem}

Information about stability is to be propagated from the MM system to the MM-MA
system by using a result showing that solutions of the MM-MA system 
approximate solutions of the MM system for $\epsilon$ small. For this it is 
useful to introduce a condensed notation for the equations
(\ref{rescaled1})-(\ref{rescaled7}) with fixed values 
of $\bar Y$, $\tilde E$ and $\tilde F$. Let 
$x=(x_Y-x_Y^*,x_{YPP}-x_{YPP}^*)$ and 
$y=(\tilde x_{YE}-\tilde x_{YE}^*,\tilde x_{YPE}-\tilde x_{YPE}^*,
\tilde x_{YPF}-\tilde x_{YPF}^*,\tilde x_{YPPF}-\tilde x_{YPPF}^*)$. Here the 
quantities $x_Y^*$ and $x_{YPP}^*$ are the coordinates of 
the stationary solution $B$ given by (\ref{expstat}) and the other starred
quantities are those obtained from $x_Y^*$ and $x_{YPP}^*$ by substituting them
into 
\eqref{stat1}-\eqref{frel}.
Then the evolution equations for the variables $x$ and $y$ in the time scale $\tau= \epsilon t$ are of the form 
(\ref{gspt}). The dependence on $\epsilon$ of the right hand 
sides of these equations arises because $\bar Y$ depends on $\epsilon$. 
For $\epsilon=0$ the  second equation of the reduced system  (\ref{gspt0}) becomes an algebraic equation 
which can be solved in the form $y=h_0(x)$ so that $g(x,h_0(x))=0$. 
Substituting this into the first equation of the reduced system  (\ref{gspt0}) gives an equation of the form 
$x'=f(x,h_0(x))$, which is  the MM-system (\ref{mm1})-(\ref{mm2}). The conditions 
$h_0(0)=0$ and $g(0,0)=0$ hold. The limit of these equations for 
$\epsilon\to 0$ can be analysed using the results of \cite{fenichel79} 
as mentioned in Sect. \ref{techniques}. Recall the form of the extended
system in terms of the time variable $t= \epsilon^{-1}\tau$:
\begin{eqnarray}
&&\dot x=\epsilon f(x,y,\epsilon),\label{reduced3}\\
&&\dot y=g(x,y,\epsilon),\label{reduced4}\\
&&\dot \epsilon=0\label{reduced5}.
\end{eqnarray}  
At each of the points of the manifold $M_0$ defined by the vanishing of
$g(x,y,0)$ the rank of the matrix which defines the linearization of the 
system is four. This is equal to the rank of the matrix defining the 
linearization of the equation $y'=g(x,y,0)$ for fixed values of the
variables $x$. In this case the calculation is particularly simple since
the right hand side of the equation is an affine function of the unknowns.
The matrix is the direct sum of two diagonal $2\times 2$
blocks. Each of these blocks has positive determinant and negative trace 
and thus all eigenvalues have negative real parts. Thus in particular at 
each of the points of $M_0$ there is a three-dimensional centre 
manifold. It follows from Theorem 9.1 of \cite{fenichel79} that there is a 
three-dimensional invariant manifold which is a centre manifold for all of 
these points. It can be chosen to be $C^k$ for any finite $k$ and it can be 
written in the form $y=h(x,\epsilon)$, where $h(x,0)=h_0(x)$. This manifold 
can be coordinatized by $x$ and $\epsilon$ and the restriction of the system 
to the manifold can be represented in the form
\begin{eqnarray}
&&x'=f(x,h(x,\epsilon)),\\
&&\epsilon'=0
\end{eqnarray} 
Since it has been shown that the system for $\epsilon=0$ has three hyperbolic 
stationary solutions it follows that the system for small positive values of 
$\epsilon$ also has at least three stationary solutions. The heteroclinics connecting the unstable equilibrium to the stable one also persist for small values of $\epsilon$.  The three  
stationary solutions are also equilibria  of the full system and it follows from the theorem of 
Shoshitaishvili that two of them are stable and the other is a saddle of Morse index one connecting to both stable equilbria. There 
are no more stationary solutions due to the results of \cite{wangsontag08a}. 
For a fixed value of $\epsilon>0$ and using the   values of 
$\bar Y$, $\tilde E$ and $\tilde F$ positive quantities
$ x_{YP}$, $\tilde x_E$ and $\tilde x_F$ can be determined which define
positive stationary solutions of the rescaled version of 
(\ref{dfc1})-(\ref{dfc9}). Undoing the rescaling gives solutions of the 
original system which all have the same values  of the conserved quantities. This 
completes the proof of Theorem \ref{bistabMM-MA}.

Since several conditions are imposed on the parameters of the system, we want 
to summarize here the remaining degrees of freedom. We already saw that the 
simplifying condition \eqref{common}, given $a_1$, $a_2$, $a_3$, $a_4$, $b$, 
$\tilde{E}$ and $\tilde{F}$, leaves 4 degrees of freedom for the coefficients 
$k_i$. The equilibrium condition \eqref{equilcond} is a condition on the 
parameters $a_1$, $a_2$, $a_3$ and $a_4$ and leave us three degrees of freedom 
for them: given $a_1$, $a_2$, $a_3$ and $a_4$ we have $a_2=a_1 a_2/a_4$. 
Putting this into the bifurcation condition \eqref{bifcond} gives $a_4<a_3$, 
which is only a slight restriction to the freedom of chosing $a_3$. Hence, the 
parameters $b$, $a_1$, $a_4<a_3$, $\tilde{E}$, $\tilde{F}$ and for example 
$k_1$, $k_4$, $k_{10}$ and $k_7$ can be chosen freely and then all other 
parameters in the system (the other $k_i$, $a_2$ and $\tilde{Y}$) are 
determined uniquely by the conditions \eqref{common}, \eqref{equilcond} and 
\eqref{bifcond}. In particular, we can choose $a_1$ and $a_4$ in such a way 
that the system is indeed in the cusp region where bistability holds. 

We can also look at those conditions from the  point of view of fixed   
  reaction constants $k_i$.  We saw 
that we can see the simplifying condition \eqref{common} as a comparison 
between reaction constants of production and consumption of 
substrate-enzyme complexes during phosphorylation. Once the simplifying 
assumption has been made, the equilibrium condition \eqref{equilcond} 
translates to 
\begin{equation} \label{quot}
\left( \frac{  \tilde{E}  }{   \tilde{F}} \right)^2= \frac{k_{12}}{k_3}.\frac{k_{9}}{k_6}.
\end{equation}
The quantities of enzymes $\tilde{E}, \tilde{F}$ have to be tuned 
accordingly.
Together with  bifurcation condition \eqref{bifcond}, equation \eqref{quot} leads to the inequality
\begin{equation}
\frac{k_{12}}{k_3}>\frac{k_{9}}{k_6}
\end{equation} 
The quotients $\frac{k_{12}}{k_3}$ and $\frac{k_{9}}{k_6}$ again compare reaction constants: 
$\frac{k_{12}}{k_3}$ compares the constants of the reactions 
\begin{equation*}
\xymatrix{
YE  \ar[r]^{k_3} & YP+E & , & YPF \ar[r]^{k_{12}} & Y+F ,
}
\end{equation*}
and the second quotient the analogous reactions for the double phosphorylated 
form $YPP$. 
In order to understand the biological meaning of this inequality it may be useful to state its concrete interpretation in terms of the reactions. 
The relative rate of production of the phosphorylated and dephosphorylated form is greater for the first cycle than for the second one. 
Furthermore the total quantity of $Y$ has to be tuned accordingly.

\section{Further directions}\label{further}

An important cellular signalling system is the MAPK cascade. It contains three
layers, each of which is a multiple phosphorylation loop of the type of the
multiple futile cycle. They are linked by the fact that the fully phosphorylated
form of the protein which is the substrate in one layer is the kinase for the 
next layer. 
\begin{equation*}
\xymatrix@C=0.5em{
 MKKK\ar@/^1pc/[rr]^{1}&&MKKK-P\ar@/^1pc/[ll]^{2}   \ar[dl]  \ar[dr]&&&&\\
 &&&&&&\\
 MKK\ar@/^2pc/[rr]^{3} &&MKK-P\ar@/^2pc/[rr]^{4} \ar@/^1pc/[ll]^{6}  &&MKK-PP\ar@/^1pc/[ll]^{5} \ar[dl]  \ar[dr] &&\\
 &&&&&&\\
 &&MAPK\ar@/^2pc/[rr]^{7}&&MAPK-P\ar@/^2pc/[rr]^{8} \ar@/^1pc/[ll]^{10}&&MAPK-PP\ar@/^1pc/[ll]^{9}  
}
\end{equation*}
Note that we later focus on the two upper layers of this cascade and use notation $X$ and $Y$ for the substrates $MKKK$ and $MKK$ respectively. 
A fundamental study, both theoretical and experimental, of this 
system was carried out by Huang and Ferrell \cite{huangferrell96}. They 
modelled it by a system of ODE using mass action kinetics. Numerical and 
heuristic evidence has been found indicating that this system has periodic 
solutions \cite{qiao07}. It may be that there is also chaotic behaviour in 
solutions of this system \cite{zumsande10}. Here we will concentrate on the
issue of oscillations. To what extent can the ideas used to prove 
bistability here be applied to give an analytical proof of oscillatory
behaviour in the MAPK cascade and can Michaelis-Menten reduction help?

MAPK cascades are often embedded in positive or negative feedback loops 
and it was suggested in \cite{kholodenko00} that negative feedback could lead
to oscillations in this context. A model was introduced where the MAPK cascade
is in a negative feedback loop belonging to a more complex network of chemical 
reactions taking place in the cell. In this model a feedback loop between a 
phosphorylated form of one protein and a substrate in a layer above is induced 
by chemical reactions external to the cascade. These other reactions are not 
modelled in detail - instead a generic form is chosen for the coupling. 
In the paper \cite{kholodenko00} a Michaelis-Menten-like description of the 
cascade was used but the assumptions leading to the equations used there were 
not stated. An insightful discussion of the modelling issues involved has been 
given in \cite{ventura08}. The equations in \cite{kholodenko00} with feedback 
loop are of the type called 'phenomenological' in \cite{ventura08}. An 
equation with mass action kinetics like the MM-MA model considered above is 
called 'mechanistic'. Finally a model like the MM model considered above is 
referred to as 'reduced mechanistic'. We will adopt the terminology of 
\cite{ventura08} in what follows. In a phenomenological model reduced 
equations for the individual reactions in a network are used and this may fail 
to capture interactions between the different reactions. 

The difference between phenomenological and reduced mechanistic models will 
now be discussed in some detail in the case of the dual futile cycle. A model 
for the dual futile cycle can be extracted from the feedback model of 
\cite{kholodenko00} as follows. We take the second layer of the cascade and 
set the concentration of the phosphorylated form of the protein in the first 
layer, $[MKKK-P]$ in the notation of \cite{kholodenko00}, to a constant 
value. This gives equations for the quantities $[MKK]$ and $[MKK-PP]$ 
corresponding to our $x_Y$ and $x_{YPP}$. For example, the equation for 
the first of these is of the form
\begin{equation}
\frac{d [MKK]}{d\tau}=v_6-v_3
\end{equation} 
where
\begin{equation}
v_3=\frac{(k_3[MKKK-P])[MKK]}{K_3+[MKK]}.
\end{equation} 
This quantity $v_3$ corresponds to $v_1$ in our notation and the key difference
is that it does not depend on $[MKK-P]$ while our $v_1$ does depend on 
$x_{YP}$. In this case the model extracted from \cite{kholodenko00} is
the phenomenological one while our model is the reduced mechanistic one. 

Returning to the question of oscillations in the MAPK cascade, note first
that according to \cite{qiao07} oscillations are already found numerically in 
a truncated version of the Huang-Ferrell model with only two layers, the first 
involving only one phosphorylation and the second involving two. For 
simplicity we now concentrate on that truncated model. 
More precisely, the reactions involved are the following:
\begin{eqnarray*}
\text{1st layer:}&
\xymatrix{
X+E \ar@/^1pc/[r]^{k_1}
&  XE \ar@/^1pc/[l]^{k_2} \ar[r]^{k_3}& XP+E \\
&&\\
XP+F \ar@/^1pc/[r]^{k_1}
&  XP.F \ar@/^1pc/[l]^{k_2} \ar[r]^{k_3}& X+F
} 
\\
\text{2nd layer:}&
\xymatrix{
Y+XP \ar@/^1pc/[r]^{k_4}
&  Y.XP \ar@/^1pc/[l]^{k_5} \ar[r]^{k_6}& YP+XP \\
&&\\
YP+XP \ar@/^1pc/[r]^{k_7} & YP.XP\ar@/^1pc/[l]^{k_8} \ar[r]^{k_9} & YPP+XP\\
&&\\
YPP+G \ar@/^1pc/[r]^{k_{10}} & YPP.G \ar@/^1pc/[l]^{k_{11}} \ar[r]^{k_{12}}& YP+G\\
&&\\
YP+F \ar@/^1pc/[r]^{k_{13}} & YP.G\ar@/^1pc/[l]^{k_{14}} \ar[r]^{k_{15}} & Y+G
}
\end{eqnarray*}

Michaelis-Menten
reduction is more subtle in this case due to the fact that a single 
substance, $XP$,  is both a substrate in one reaction and an enzyme in others.
This problem has been analysed in \cite{ventura08} and \cite{ventura13}.
In fact these papers treat more general sequences of coupled multiple
phosphorylation loops. The numerical results of \cite{ventura13} lead to
the conclusion that the periodic solutions of the mechanistic model seen in \cite{qiao07} are also
present in the reduced mechanistic model. They also show that the 
reduced mechanistic model often gives more accurate results than a 
phenomenological model even when the parameter $\epsilon$ is not very small. 

While the reduction procedure of \cite{ventura08} and \cite{ventura13} is 
close to being a direct generalization of the reduction to the MM system done
above to the case of the truncated Huang-Ferrell model it is not exactly
the same and it needs to be modified slightly to give a limit for 
$\epsilon\to 0$ which is well-behaved from the point of view of GSPT.
This modification will now be carried out. The notation is chosen to fit
that in the rest of the paper. The proteins in the first and second layers
are denoted by $X$ and $Y$ respectively. The kinase in the first layer is 
denoted by $E$ and the phosphatases in the first and second layers by $F$
and $G$. The basic unknowns are the concentrations of $X$, $XP$, $Y$, $YP$, 
$YPP$, $XE$, $XP.F$, $Y.XP$, $YP.XP$, $YP.G$, $YPP.G$, $E$, $F$ and $G$. 
The conservation laws for the three enzymes $E$, $F$ and $G$ and the two
substrates $X$ and $Y$ can be used to eliminate five of the variables. A
key idea in \cite{ventura08} is to introduce a new variable 
$x_*=x_{XP}+x_{Y.XP}$, which leads to a cancellation in one of the evolution 
equations. New variables are introduced by rescaling in a way similar to that 
done for the dual futile cycle above. They are defined by the relations
$x_*=\epsilon\tilde x_*$, $x_{EX}=\epsilon^2\tilde x_{EX}$, 
$x_{F.XP}=\epsilon^2\tilde x_{F.XP}$, $x_{XP.Y}=\epsilon\tilde x_{XP.Y}$, 
$x_{XP.YP}=\epsilon\tilde x_{XP.YP}$, $x_{G.YP}=\epsilon\tilde x_{G.YP}$, 
$x_{G.YPP}=\epsilon\tilde x_{G.YPP}$, $x_E=\epsilon^2\tilde x_E$,
$x_F=\epsilon^2\tilde x_F$, $x_G=\epsilon\tilde x_G$, $\tau=\epsilon t$.
In addition, and this is the difference to what was done in \cite{ventura08},
two of the reaction constants are rescaled, $\tilde k_1=k_1\epsilon$ and 
$\tilde k_4=k_4\epsilon$. These are the reaction constants for the binding of
the enzymes $E$ and $F$ to their substrates. The result is the system
\begin{eqnarray}
&&\frac{d\tilde x_*}{d\tau}=k_3 \tilde x_{EX}
-\tilde k_4 (\tilde x_*-\tilde x_{XP.Y}-\tilde x_{XP.YP})\tilde x_F
+k_5 \tilde x_{F.XP}\\
&&\frac{d\tilde x_{Y}}{d\tau}=
-k_7 \tilde x_{Y}(\tilde x_*-\tilde x_{XP.Y}-\tilde x_{XP.YP})
+k_8 \tilde x_{XP.Y}+k_{12}\tilde x_{G.YP},\\
&&\frac{d\tilde x_{YPP}}{d\tau}=k_{15}\tilde x_{XP.YP}
-k_{16}\tilde x_{YPP}\tilde x_G+k_{17}\tilde x_{G.YPP},\\
&&\epsilon\frac{d\tilde x_{EX}}{d\tau}=\tilde k_1 \tilde x_{X}\tilde x_{E}
-k_2 \tilde x_{EX}-k_3 \tilde x_{EX},\\
&&\epsilon\frac{d\tilde x_{F.XP}}{d\tau}=
\tilde k_4 (\tilde x_*-\tilde x_{XP.Y}-\tilde x_{XP.YP})\tilde x_F
-k_5 \tilde x_{F.XP}-k_6 \tilde x_{F.XP},\\
&&\epsilon\frac{d\tilde x_{XP.Y}}{d\tau}=
k_7 \tilde x_{Y}(\tilde x_*-\tilde x_{XP.Y}-\tilde x_{XP.YP})-k_8 \tilde x_{XP.Y}
-k_9 \tilde x_{XP.Y},\\
&&\epsilon\frac{d\tilde x_{XP.YP}}{d\tau}=
k_{13}\tilde x_{YP}(\tilde x_*-\tilde x_{XP.Y}-\tilde x_{XP.YP})
-k_{14}\tilde x_{XP.YP}-k_{15}\tilde x_{XP.YP},\\
&&\epsilon\frac{d\tilde x_{G.YP}}{d\tau}=k_{10}\tilde x_{YP}\tilde x_G
-k_{11}\tilde x_{G.YP}-k_{12}\tilde x_{G.YP},\\
&&\epsilon\frac{d\tilde x_{G.YPP}}{d\tau}=k_{16}\tilde x_{YPP}\tilde x_G
-k_{17}\tilde x_{G.YPP}-k_{18}\tilde x_{G.YPP}.
\end{eqnarray}
It is then possible to pass to the limit $\epsilon\to 0$. Carrying out  
steps analogous to those done in the case of the dual futile cycle 
in a routine way gives
\begin{eqnarray}
&&\frac{d\tilde x_*}{d\tau}=\frac{c_1\tilde x_X}{1+d_1\tilde x_X}
-\frac{c_2\tilde x_*}{1+d_2\tilde x_*
+b_1\tilde x_Y+b_2\tilde x_{YP}},\label{mapkmm1}\\
&&\frac{d\tilde x_{Y}}{d\tau}=-\frac{a_1\tilde x_*\tilde x_Y}
{1+b_1\tilde x_Y+b_2\tilde x_{YP}}
+\frac{a_2\tilde x_{YP}}{1+b_3\tilde x_{YP}+b_4\tilde x_{YPP}},\label{mapkmm2}\\
&&\frac{d\tilde x_{YPP}}{d\tau}=\frac{a_3\tilde x_*\tilde x_{YP}}
{1+b_1\tilde x_Y+b_2\tilde x_{YP}}
-\frac{a_4\tilde x_{YPP}}{1+b_3\tilde x_{YP}+b_4\tilde x_{YPP}}.\label{mapkmm3}
\end{eqnarray}
The constants in this system can be expressed explicitly in terms of
the parameters in the full system. To get a closed system the conservation
laws for the total amounts of the substrates $X$ and $Y$ must be used. 

This system may be compared with one extracted from the system of
\cite{kholodenko00} in a way similar to that done for the dual futile cycle 
above. This time we take the equations for the substances in the first two
layers and remove the explicit feedback. 
This is done by setting the concentrations $[MAPK]$, $[MAPK-P]$ and $[MAPK-PP]$ to zero. Then the equations for those quantities are satisfied identically and the dependence of the first reaction rate $v_1$ in \cite{kholodenko00} on $[MAPK-PP]$ disappears. 
In the resulting system the evolution equations for the second layer substrates
$[MKK]$ and $[MKK-PP]$ are independent of the variables involving the  third layer substrate $[MAPK]$. This 
is in contrast to the system (\ref{mapkmm1})-(\ref{mapkmm3}) where there is
an intrinsic negative feedback loop. For if the right hand sides of these
equations are denoted by $f_i, i=1,2,3,$ then 
$\partial f_1/\partial\tilde x_{YPP}<0$ while 
$\partial f_3/\partial\tilde x_*>0$. Multiplying the two different signs of 
these derivatives gives the claimed negative feedback. (For background on
the defintion and significance of feedback loops the reader is referred to
\cite{sontag07}.) In the equations
extracted from the system of \cite{kholodenko00}, on the other hand, we have
$\partial f_1/\partial[MK-PP]=0$ and there is no intrinsic feedback. 
According to common heuristics the negative feedback loop in 
(\ref{mapkmm1})-(\ref{mapkmm3}) can be seen as a possible source of 
oscillations. The interaction not captured in the phenomenological description 
is what is called sequestration. In certain circumstances a lot of the substance
$XP$ is bound to its substrates $Y$ and $YP$ and so is not available to
take part in the reaction converting it back to $X$ in the first layer.

The equations (\ref{mapkmm1})-(\ref{mapkmm3}) are a reduced system for 
the truncated Huang-Ferrell model in the sense of GSPT. The transverse 
eigenvalues (as defined in \ref{transeigen}) have negative real parts. In contrast to the system arising
for the dual futile cycle the right hand side of the equations entering 
the computation of the transverse eigenvalues is no
longer affine. However the only nonlinear terms are in the equation for
$\tilde x_{F.XP}$ and the derivatives of these terms do not influence the
eigenvalues. All the other terms are either diagonal or belong to $2\times 2$
blocks which are similar to those already treated in the case of the dual 
futile cycle. Hence if it were possible to show that 
(\ref{mapkmm1})-(\ref{mapkmm3}) has a hyperbolic stable periodic solution 
then it would follow immediately that for $\epsilon$ sufficiently small the 
full system also has a periodic solution. This is because hyperbolic stable 
periodic solutions persist under regular perturbations of a dynamical system. 
Thus the task of proving the existence of the periodic solutions of the 
truncated Huang-Ferrell model found numerically in \cite{qiao07} reduces to 
proving an analogous statement for the reduced system. Since the former
has dimension nine, even after exploiting all conservation laws, and the
latter has dimension three, this could be a significant advantage. The 
question of existence of periodic solutions of the reduced system will 
be investigated elsewhere. 

Another interesting question concerns the dependence of the results of 
Theorems \ref{bistabMM} and \ref{bistabMM-MA} on the modelling by 
mass-action kinetics. A crucial argument in the present paper consists in 
verifying the conditions for a cusp bifurcation to take place. It is already 
technical for the mass-action kinetics where equilibria and derivatives can be 
computed explicitly. On the one hand, such computations seem difficult with a 
more general model for the kinetics of the chemical reactions. But on the 
other hand, the fact that the evolution equations have a polynomial reaction rates is not essential. Models with more general right hand side could 
lead to similar results. Furthermore, we conjecture that the assumption 
\eqref{common} is not a necessary condition, but only a convenience to 
simplify the computations.

\vskip 20pt\noindent
{\it Acknowledgements} ADR thanks David Angeli, Carsten Conradi, Stefan
Legewie, Jacques Sepulchre and Peter Szmolyan for helpful conversations and 
correspondence.

\end{document}